\documentclass[a4paper,11pt]{quantumarticle}
\pdfoutput=1
\usepackage[utf8]{inputenc}
\usepackage[english]{babel}
\usepackage[T1]{fontenc}
\usepackage{amsmath,amssymb,amsfonts}
\usepackage{amsthm}
\usepackage{booktabs}
\usepackage{graphicx}
\usepackage{bm}
\usepackage{hyperref}

\newtheorem{lemma}{Lemma}
\newtheorem{proposition}{Proposition}
\theoremstyle{definition}

\newcommand{\F}{\mathbb{F}}
\newcommand{\vx}{\boldsymbol{x}}
\newcommand{\ve}{\boldsymbol{e}}
\newcommand{\vs}{\boldsymbol{s}}
\newcommand{\vu}{\boldsymbol{u}}
\newcommand{\vr}{\boldsymbol{r}}
\newcommand{\vlam}{\boldsymbol{\lambda}}
\newcommand{\vtheta}{\boldsymbol{\theta}}
\newcommand{\Z}{\mathrm{Z}}
\newcommand{\PP}{\mathbb{P}}

\begin{document}

\title{Certified decoding of quantum LDPC codes}

\author{Ragavi Krishnamoorthy}
\affiliation{Hybrid Intelligence, Fraunhofer IAIS}
\author{Florian Gerhardt}
\author{Johannes Knaute}
\affiliation{Actuarial Risk Modeling, Quantum Computing Algorithms \& Applications, PricewaterhouseCoopers GmbH} %Wirtschaftspr\"ufungsgesellschaft}
\author{Thomas Klir}
\affiliation{Cyber Security, PricewaterhouseCoopers GmbH} %Wirtschaftspr\"ufungsgesellschaft}
\author{Stefan Raimund Maschek}
\author{Erik Schulze}
\affiliation{Quantum Computing, Bull GmbH}
\author{Tomislav Maras}
\author{Alexander Dotterweich}
\affiliation{Actuarial Risk Modeling, Quantum Computing Algorithms \& Applications, PricewaterhouseCoopers GmbH} %Wirtschaftspr\"ufungsgesellschaft}
\author{Loong Kuan Lee}
\author{Christian Bauckhage}
\author{Nico Piatkowski}
\affiliation{Hybrid Intelligence, Fraunhofer IAIS}

\begin{abstract}
Quantum low-density parity-check (qLDPC) codes reduce the qubit overhead of
fault-tolerant quantum computation by an order of magnitude, but their
decoding is harder than its classical counterpart: because many physical
errors are equivalent up to stabilizers, the degenerate maximum-likelihood
(ML) decoder must compare the probabilities of entire \emph{equivalence
classes} of errors, that is, partition functions, rather than single errors. The
workhorse decoder BP+OSD sidesteps degeneracy heuristically and offers no
guarantees. We treat degenerate decoding as probabilistic inference in an
undirected graphical model: the probability of each logical class is the
partition function of an \emph{unconstrained, strictly positive} Markov
random field over the code's check variables, a construction that
generalizes the random-bond Ising mapping of the surface code to arbitrary
CSS codes and to spacetime decoding with measurement errors and
circuit-level noise. On this model we build two decoders. The first
estimates all class partition functions by annealed importance sampling
with common random numbers and attaches to every decision a
\emph{certificate} of optimality: a paired bootstrap test, or, composed
with constant-factor estimators such as WISH, an exact optimality proof.
The second is region-based: the \emph{Bethe free energy}, whose bias
cancels between classes, reproduces exact ML decoding on every tested
surface-code instance at millisecond cost, and enlarging the regions to
elimination clusters makes exact degenerate ML decoding of the
$[\![72,12,6]\!]$ bivariate bicycle code feasible. Across surface codes
and the bivariate bicycle codes $[\![72,12,6]\!]$ and
$[\![144,12,12]\!]$, under code-capacity, phenomenological, and
circuit-level noise, the sampling decoder matches or exceeds BP+OSD while
certifying the bulk of its decisions, and the certificate flags exactly
the syndromes on which any fast decoder should be distrusted.
\end{abstract}

\maketitle

\section{Introduction}
\label{sec:intro}
Large-scale quantum computation requires quantum error correction (QEC),
because physical qubit error rates ($\sim 10^{-3}$--$10^{-2}$ in current
hardware) exceed the requirements of useful algorithms by many orders of
magnitude. For over a decade, the surface code~\cite{kitaev2003, dennis2002,
fowler2012} has been the de facto standard: it tolerates comparatively high
physical error rates ($\sim 1\%$ threshold) and requires only nearest-neighbor
interactions on a two-dimensional (2D) qubit lattice. Its main drawback is
overhead: a single well-protected logical qubit may consume on the order of
$10^3$ physical qubits. Quantum low-density parity-check (qLDPC) codes attack
precisely this bottleneck. By encoding many logical qubits into a single
sparse code block, they promise a reduction of the physical-to-logical qubit
ratio by an order of magnitude or more~\cite{gottesman2014, breuckmann2021,
panteleev2022good, leverrier2022}, with the bivariate bicycle (BB) codes of
Bravyi \emph{et al.}~\cite{bravyi2024} as the currently most prominent
hardware-motivated construction.

The price of dense encoding is paid, in part, at the decoder.
Decoding a qLDPC code is markedly harder than decoding either classical LDPC
codes or the surface code, for a reason that is structural rather than
incidental: quantum codes are \emph{degenerate}. Errors that differ by a
stabilizer act identically on the code space, so the optimal
\emph{degenerate maximum-likelihood} (ML) decoder must identify the most
probable \emph{equivalence class} of errors, a sum over exponentially many
configurations, rather than the most probable single error. The de facto
standard decoder, belief propagation with ordered-statistics post-processing
(BP+OSD)~\cite{panteleev2021bposd, roffe2020}, handles degeneracy
heuristically: it estimates single-error marginals by loopy message passing
and repairs their inconsistencies by a greedy matrix-elimination step. It
performs well in practice but offers no optimality statement, no error bound,
and no indication of which of its decisions are trustworthy.

This paper develops a decoder that takes the structure of the problem
seriously. Our starting point is the observation, made precise in
Section~\ref{sec:pgm}, that syndrome decoding \emph{is} probabilistic
inference in an undirected graphical model (a Markov random field, MRF)
whose factor graph is the code's Tanner graph, and that degeneracy moves the
relevant inference task from maximum-a-posteriori estimation to
\emph{partition-function evaluation}. Partition-function estimation is a
mature subfield of probabilistic machine learning, with estimators whose
statistical behavior is well understood~\cite{neal2001, ermon2013,
wainwright2008}. We import this machinery into qLDPC decoding.

\paragraph{Contributions.}
\begin{itemize}
\item \emph{The coset MRF.} The probability of each logical class is the
partition function of an \emph{unconstrained, strictly positive} MRF over
the code's $X$-check variables (Proposition~\ref{prop:cosetmrf}). The
construction applies to any CSS code and generalizes the classical mapping
of surface-code decoding to the random-bond Ising
model~\cite{dennis2002} (for the surface code it reproduces exactly that
Ising model, for BB codes it yields a three-body analogue), and it extends
verbatim to spacetime decoding problems with faulty measurements and
circuit-level noise (Section~\ref{sec:spacetime}). Crucially, the hard
parity constraints of the syndrome posterior are eliminated by the
parametrization, so every standard sampler and message-passing scheme
applies without modification.
\item \emph{A certified sampling decoder.} Annealed importance sampling
(AIS)~\cite{neal2001} with \emph{common random numbers}
(CRN)~\cite{robert2004} across classes estimates all class partition
functions on one shared randomness stream, turning the decision into a
\emph{paired} comparison. A paired bootstrap test on the coupled estimates
certifies each decision; composed with constant-factor estimators such as
WISH~\cite{ermon2013}, whose optimization oracles are QUBOs, hence native
to quantum annealers~\cite{kadowaki1998}, the certificate becomes an exact
optimality proof (Lemma~\ref{lem:cert}). We also report two natural
``sharper'' ratio estimators that \emph{fail} instructively: direct
inter-class annealing and Bennett's acceptance ratio both collapse on the
low-overlap coset pairs, pinning down why the anneal through infinite
temperature is the right bridge (Section~\ref{sec:cert}).
\item \emph{Region-based decoders, from Bethe to exact.} The Bethe free
energy of each coset MRF, computed by loopy belief propagation in a tanh
parametrization, is a biased estimate of $\log\Z_{\vlam}$, but the bias
is shared across classes and cancels in the comparison: the Bethe decoder
reproduces the exact ML decision on every tested surface-code instance at
millisecond cost (Section~\ref{sec:bethe}). One level up, Kikuchi-size
regions realized by (mini-)bucket elimination~\cite{dechter2003} exploit a
structural windfall: the coset graph of the $[\![72,12,6]\!]$ BB code has
induced width $\approx23$, so \emph{exact degenerate ML decoding of a
qLDPC code} becomes feasible, and bounded clusters yield upper bounds
that, paired with configuration lower bounds, give \emph{deterministic}
optimality certificates and provable class pruning
(Section~\ref{sec:regions}).
\item \emph{Validation from code capacity to circuit level.} We evaluate
against \emph{exact} degenerate ML decoding on rotated surface codes, and
against BP+OSD~\cite{panteleev2021bposd, roffe2020} and
MWPM~\cite{higgott2022} on the BB codes $[\![72,12,6]\!]$ and
$[\![144,12,12]\!]$~\cite{bravyi2024} and on spacetime problems
(phenomenological noise for the BB code and stim-generated circuit-level
noise for the surface code), reporting logical error rates, certified
fractions, calibration of the certificates against ground truth, and
runtimes (Section~\ref{sec:experiments}).
\end{itemize}

Related approaches to degeneracy-aware decoding exist for the surface code,
where 2D structure makes coset sums accessible to tensor-network
contraction~\cite{bravyi2014ml} and to tailored Metropolis
dynamics~\cite{hutter2014}. Our construction requires no geometric
structure, only sparsity, and therefore covers the expander-like codes at
the center of current fault-tolerance road maps; its guarantees come from
the estimator, not from the code family.

The paper is self-contained. Sections~\ref{sec:classical}--\ref{sec:qldpc}
review the three ingredients (classical LDPC codes, undirected graphical
models, and the qLDPC program), and Section~\ref{sec:challenges} explains
why decoding is the binding constraint. Sections~\ref{sec:method}
and~\ref{sec:experiments} develop and evaluate the proposed decoders.
Section~\ref{sec:discussion} concludes.

\section{Classical LDPC Codes}
\label{sec:classical}
A binary linear code $\mathcal{C} \subseteq \F_2^n$ is defined by a
parity-check matrix $H \in \F_2^{m \times n}$ via
$\mathcal{C} = \{ \vx \in \F_2^n : H\vx = 0 \}$. The code has parameters
$[n, k, d]$, where $k = n - \mathrm{rank}(H)$ is the number of encoded bits
and $d$ is the minimum Hamming weight of a nonzero codeword. A code family is
called \emph{low-density parity-check} (LDPC) if $H$ is sparse: every row and
every column contains at most a constant number of ones, independent of $n$.
Equivalently, in the bipartite \emph{Tanner graph} connecting bit nodes to
check nodes, all node degrees are bounded by a constant.

LDPC codes were introduced by Gallager in 1962~\cite{gallager1962} and
rediscovered in the 1990s~\cite{mackay1999}, when it became clear that random
sparse codes decoded with iterative message passing (belief propagation, BP)
operate remarkably close to the Shannon limit~\cite{richardson2001}. Two
properties drive their practical dominance in modern communication standards
(e.g., 5G-NR and IEEE 802.11): (i)~sparsity makes each decoding iteration
cheap and massively parallelizable, which is why LDPC decoders are a classic
target for FPGA and ASIC implementation; and (ii)~random sparse constructions
achieve linear distance $d = \Theta(n)$ at constant rate $k/n$, i.e., they
are asymptotically \emph{good}.

The central question of the qLDPC program is whether these two properties,
sparse checks and good asymptotic parameters, can coexist in the quantum
setting, where the answer turned out to be far less obvious.

\section{Undirected Probabilistic Graphical Models}
\label{sec:pgm}
Iterative decoding is best understood as \emph{probabilistic inference}. This
section fixes notation for undirected probabilistic graphical models, also
called Markov random fields (MRFs)~\cite{pearl1988, koller2009,
wainwright2008}, in which the decoding problems appearing throughout this
paper are instances of standard inference tasks.

\subsection{Markov Random Fields and Inference}
An MRF over variables $X_1,\dots,X_N$ with $X_v\in\{0,\dots,Y_v-1\}$ and
clique set $\mathcal{C}$ ($m=|\mathcal{C}|$) factorizes as
$\PP(\vx)=\tfrac1\Z\prod_{C\in\mathcal{C}}\psi_C(\vx_C)$ with clique
potentials $\psi_C\ge0$, where $\vx_C$ denotes the restriction of $\vx$ to
the variables in $C$. Writing $\psi_C=\exp(\vtheta_C)$ in log-domain, the
\emph{partition function}
\begin{equation}
\Z=\sum_{\vx}\prod_{C\in\mathcal{C}}\psi_C(\vx_C)
 =\sum_{\vx}\exp\!\Big(\textstyle\sum_{C\in\mathcal{C}}\vtheta_C(\vx_C)\Big)
\label{eq:Z}
\end{equation}
normalizes the distribution. The three canonical inference tasks are
(i)~\emph{marginal} inference, computing $\PP(x_v)$ or, more generally, the
probability of a subset of variables by summing out the rest;
(ii)~\emph{maximum-a-posteriori} (MAP) inference,
$\vx^\star=\arg\max_{\vx}\PP(\vx)$; and (iii)~evaluation of $\Z$ itself.
All three are computationally hard in general: computing $\Z$ is
\#P-hard, and even approximate marginal inference is intractable in the
worst case~\cite{roth1996}. Exact inference via the junction-tree algorithm
runs in time exponential in the \emph{treewidth} of the underlying
graph~\cite{lauritzen1988}, which confines it to suitably sparse, tree-like
models; beyond that regime one resorts to Monte-Carlo
sampling~\cite{geman1984} or to variational
approximations~\cite{wainwright2008}.

\subsection{Factor Graphs and Belief Propagation}
The factorization is conveniently drawn as a \emph{factor graph}: a
bipartite graph with one variable node per $X_v$, one factor node per clique
$C$, and an edge whenever $X_v\in C$~\cite{kschischang2001}. Belief
propagation (BP)~\cite{pearl1988} computes marginals by passing messages
along its edges,
\begin{align}
\mu_{C\to v}(x_v) &= \sum_{\vx_{C\setminus v}}\psi_C(\vx_C)
  \prod_{u\in C\setminus v}\mu_{u\to C}(x_u),
\label{eq:bp}\\
\mu_{v\to C}(x_v) &= \prod_{C'\ni v,\,C'\neq C}\mu_{C'\to v}(x_v),
\nonumber
\end{align}
with marginals obtained as normalized products of incoming messages. On
tree-structured factor graphs BP is exact; on graphs with cycles, iterating
\eqref{eq:bp} to a fixed point (\emph{loopy} BP) yields approximate
marginals that correspond to stationary points of the Bethe free
energy~\cite{yedidia2005}. Replacing the sum in~\eqref{eq:bp} by a
maximization gives max-product BP, the analogous approximation for MAP
inference.

\subsection{Decoding as Inference}
\label{sec:decinf}
Syndrome decoding is precisely such an inference problem. Let
$\ve\in\F_2^n$ be the error of a memoryless channel with flip
probability $p$, so that the prior factorizes into unary potentials
$\psi_v(e_v)=p^{e_v}(1-p)^{1-e_v}$, and let $\vs=H\ve$ be the observed
syndrome. Conditioning on $\vs$ attaches to each check $c$ the hard factor
$\psi_c(\ve_c)=\mathbb{1}[\bigoplus_{v\in c}e_v=s_c]$, yielding the
posterior
\begin{equation}
\PP(\ve\mid\vs)=\frac1{\Z(\vs)}\prod_{v=1}^{n}\psi_v(e_v)
\prod_{c=1}^{m}\psi_c(\ve_c),
\label{eq:posterior}
\end{equation}
an MRF whose factor graph \emph{is} the Tanner graph of the code. Bitwise
MAP decoding ($\arg\max_{e_v}\PP(e_v\mid\vs)$ for each $v$) is marginal
inference in~\eqref{eq:posterior}, and iterative LDPC decoding, from
Gallager's original algorithm~\cite{gallager1962} onward, is loopy BP on
this model~\cite{mceliece1998}; blockwise MAP decoding is MAP inference.
The sparsity that defines LDPC codes is exactly the sparsity that makes
each BP iteration cheap.

The quantum setting adds a characteristic twist. Because errors that differ
by a stabilizer act identically on the code space, the optimal
(\emph{degenerate maximum-likelihood}) decoder does not seek the single most
probable error but the most probable \emph{equivalence class}: it maximizes
the total posterior mass $\sum_{\ve\in\bar\ve+\mathcal{S}}\PP(\ve\mid\vs)$
of a coset, a sum over exponentially many configurations, that is, a
constrained \emph{partition function}~\eqref{eq:Z} rather than a point
estimate. For the surface code this coset sum is the partition function of a
random-bond Ising model~\cite{dennis2002}. Degeneracy thus moves quantum
decoding from MAP inference to the strictly harder task of
partition-function evaluation, a fact that resurfaces as a central decoding
obstacle in Section~\ref{sec:challenges} and that our method
(Section~\ref{sec:method}) confronts head-on.

\section{Stabilizer Codes and the CSS Construction}
\label{sec:stabilizer}
Quantum error correction must protect against both bit-flip ($X$) and
phase-flip ($Z$) errors without measuring, and thereby destroying, the
encoded state. The stabilizer formalism~\cite{gottesman1997} achieves this by
defining the code space as the joint $+1$ eigenspace of an abelian group
$\mathcal{S}$ of $n$-qubit Pauli operators. Measuring a generating set of
$\mathcal{S}$ yields a classical \emph{syndrome} that reveals information
about errors while leaving the logical state intact. An
$[\![n, k, d]\!]$ stabilizer code encodes $k$ logical qubits into $n$
physical qubits; $d$ is the minimum weight of a Pauli operator that commutes
with all stabilizers without being a stabilizer itself, i.e., of a nontrivial
logical operator.

The Calderbank--Shor--Steane (CSS) construction~\cite{calderbank1996,
steane1996} builds a stabilizer code from two classical codes with
parity-check matrices $H_X$ and $H_Z$ satisfying the orthogonality condition
\begin{equation}
  H_X H_Z^{\mathsf{T}} = 0 ,
  \label{eq:css}
\end{equation}
which guarantees that all $X$-type and $Z$-type stabilizers commute. The
$X$-checks detect phase-flip errors and the $Z$-checks detect bit-flip
errors, so decoding partially reduces to two classical decoding problems,
with the important caveat of \emph{degeneracy}: many distinct physical errors
are equivalent up to stabilizers and need not be distinguished, which
classical decoders are not designed to exploit.

A stabilizer code family is \emph{qLDPC} if every stabilizer generator acts
on at most a constant number of qubits and every qubit participates in at most
a constant number of generators~\cite{breuckmann2021}. Sparsity is even more
critical here than classically: syndrome extraction is an active quantum
circuit, and high-weight stabilizers require deep measurement circuits that
themselves spread errors, degrading fault tolerance.

Constraint~\eqref{eq:css} is precisely what makes good quantum codes hard to
find. Random sparse matrices, the workhorse of classical LDPC theory,
essentially never satisfy it; qLDPC constructions therefore require heavy
algebraic or topological structure, and this structure tends to depress
either the rate or the distance.

\section{From the Surface Code to Quantum LDPC Codes}
\label{sec:qldpc}

\subsection{The Surface Code as a qLDPC Code}
\label{sec:surface}
The surface code~\cite{kitaev2003, dennis2002, fowler2012} is itself a CSS
qLDPC code: all stabilizers have weight four and each qubit participates in
at most four checks. Its Tanner graph is, moreover, \emph{geometrically
local} in 2D, which is the key to its popularity on planar superconducting
chips. It also enjoys a high circuit-level threshold near $1\%$ and an
efficient, well-understood decoder (minimum-weight perfect matching on the
syndrome graph).

Its parameters, however, are poor: a distance-$d$ surface code patch uses
$n = \Theta(d^2)$ physical qubits to encode $k = 1$ logical qubit, i.e.,
\begin{equation}
  k = 1, \qquad d = \Theta(\sqrt{n}), \qquad \frac{k}{n} \to 0 .
\end{equation}
Suppressing logical error rates to the $10^{-10}$--$10^{-12}$ regime needed
for large algorithms requires $d \approx 25$--$35$, hence roughly
$10^3$ physical qubits per logical qubit before accounting for magic-state
distillation. This vanishing rate is not an accident of the construction:
Bravyi, Poulin, and Terhal proved that \emph{any} 2D-local stabilizer code
obeys
\begin{equation}
  k d^2 = O(n),
  \label{eq:bpt}
\end{equation}
so the surface code is essentially optimal within its locality
class~\cite{bravyi2010}. Better parameters therefore \emph{require}
long-range connectivity.

\subsection{Asymptotically Good qLDPC Codes}
\label{sec:good}
For nearly twenty years after the CSS construction, the best known qLDPC
distance scaling barely exceeded $\sqrt{n}$. The \emph{hypergraph product} of
Tillich and Z\'emor~\cite{tillich2014} was a milestone: it turns any pair of
classical codes into a CSS code satisfying~\eqref{eq:css} by construction and
yields constant-rate families with $d = \Theta(\sqrt{n})$: constant rate,
but still square-root distance. Gottesman showed that constant-rate qLDPC
codes would enable fault-tolerant computation with \emph{constant} space
overhead~\cite{gottesman2014}, sharpening the motivation to push distance
further.

The breakthrough came in rapid succession. Panteleev and Kalachev introduced
\emph{lifted product} codes, first achieving almost-linear
distance~\cite{panteleev2021distance} and then proving the existence of
asymptotically good qLDPC codes with $k = \Theta(n)$ and $d = \Theta(n)$
via lifted products of Tanner codes over non-abelian group
algebras~\cite{panteleev2022good}. Shortly thereafter, Leverrier and Z\'emor
gave an elegant alternative construction, the \emph{quantum Tanner
codes}~\cite{leverrier2022}, based on expander graphs equipped with local
codes on squares of a group. Conceptually, both constructions replace
geometric locality with \emph{expansion}: the Tanner graph is sparse yet
highly connected, so that small errors cannot conspire to form low-weight
logical operators. The price is exactly the long-range connectivity
forbidden by~\eqref{eq:bpt}~\cite{baspin2022}. The good constructions are,
at present, primarily existence results, which motivates the
intermediate-scale constructions of the next subsection.

\subsection{Bivariate Bicycle Codes}
\label{sec:bb}
The bivariate bicycle (BB) codes of Bravyi \emph{et
al.}~\cite{bravyi2024} are the currently most prominent attempt to bring
qLDPC advantages into the practically relevant regime of hundreds of qubits.
BB codes are CSS codes defined by two sparse polynomials $A$ and $B$ over the
group algebra of $\mathbb{Z}_\ell \times \mathbb{Z}_m$, with check matrices
$H_X = [A \mid B]$ and $H_Z = [B^{\mathsf{T}} \mid A^{\mathsf{T}}]$, which
satisfies~\eqref{eq:css} because $A$ and $B$ commute. All stabilizers have
weight six, and each qubit couples to six others, only marginally more
demanding than the surface code's degree four.

The flagship instance is the $[\![144, 12, 12]\!]$ ``gross''
code\footnote{The name derives from the unit \emph{gross} $=144$.}: 144 data
qubits (plus 144 ancilla qubits for syndrome extraction) encode 12 logical
qubits at distance 12. Circuit-level simulations show a pseudo-threshold near
$0.7\%$, and at a physical error rate of $10^{-3}$ the code reaches logical
error rates for which a surface-code implementation of equal performance
would require roughly ten times as many physical qubits~\cite{bravyi2024}.
Crucially, the required connectivity, while not 2D-local, is
\emph{almost} planar: the Tanner graph decomposes into two planar layers plus
a bounded number of long-range couplers per qubit, a topology considered
realistic for superconducting chips with through-substrate vias or
long-range resonators. The headline factor-of-ten saving applies to
\emph{quantum memory}; the comparison becomes less clear-cut once logical
operations are included, and the twelve logical qubits of a gross-code block
share correlated failure mechanisms.

\section{The Decoding Bottleneck}
\label{sec:challenges}
\label{sec:decoding}
Deploying qLDPC codes faces three obstacles. Two of them are outside the
scope of this paper and we only note them: the long-range \emph{connectivity}
mandated by~\eqref{eq:bpt} must be engineered or provided
natively~\cite{bravyi2010, baspin2022, cohen2022, xu2024}, and
fault-tolerant \emph{logical gates} on densely packed code blocks are far
less mature than surface-code lattice surgery~\cite{cohen2022, bravyi2024}.
The third obstacle is decoding, and it is the subject of everything that
follows.

Decoding qLDPC codes is markedly harder than decoding either classical LDPC
codes or the surface code. Plain BP fails on quantum Tanner graphs for two
reasons: unavoidable short cycles (four-cycles are built into CSS
constructions) and error degeneracy, which splits the posterior probability
mass over many equivalent errors and prevents BP from converging to any one
of them. In the language of Section~\ref{sec:pgm}: the syndrome-conditioned
posterior~\eqref{eq:posterior} has many symmetric modes, so its sum-product
marginals are uninformative, and the quantity a degenerate decoder must
compare is a coset partition function rather than a MAP configuration.
The current standard is BP with ordered-statistics post-processing
(BP+OSD)~\cite{panteleev2021bposd, roffe2020}: when BP fails to converge, an
OSD step solves a most-likely-error problem on a reliability-sorted, reduced
system via Gaussian elimination. BP+OSD performs well across the qLDPC
landscape but its post-processing has cubic worst-case cost, which conflicts
with the real-time requirement that syndrome information be processed within
microseconds on superconducting hardware, and it comes with no optimality
or error guarantee whatsoever.

\section{Degenerate ML Decoding as Partition-Function Estimation}
\label{sec:method}
We now develop the proposed decoders. Throughout, we consider a CSS code with
check matrices $H_X\in\F_2^{m_X\times n}$, $H_Z\in\F_2^{m_Z\times n}$ and
independent bit-flip ($X$-type) noise of rate $p<\tfrac12$; the phase-flip
problem is identical with the roles of $H_X$ and $H_Z$ exchanged.

\subsection{The Decision Problem}
\label{sec:decision}
An error $\ve\in\F_2^n$ with prior
$\PP(\ve)=(1-p)^n w^{|\ve|}$, $w = p/(1-p)$, produces the syndrome
$\vs=H_Z\ve$. Fix any particular solution $\ve_0$ of $H_Z\ve_0=\vs$ and a
basis $L_X\in\F_2^{k\times n}$ of logical $X$ representatives (a basis of
$\ker H_Z$ modulo the rowspace of $H_X$). The solution set of the syndrome
equation decomposes into $2^k$ logical classes with representatives
$\ve_{\vlam}=\ve_0\oplus\vlam^{\mathsf{T}}L_X$, $\vlam\in\F_2^k$, each class
being the coset $\ve_{\vlam}\oplus\mathrm{rowspace}(H_X)$. The degenerate ML
decoder returns
\begin{equation}
\hat{\vlam}=\arg\max_{\vlam\in\F_2^k}\Z_{\vlam}(\vs),
\;\;
\Z_{\vlam}(\vs)=\sum_{\vu\in\F_2^{m_X}}
   w^{|\ve_{\vlam}\oplus\vu^{\mathsf{T}}H_X|},
\label{eq:cosetZ}
\end{equation}
and any representative of the winning class is an optimal correction: the
decoder succeeds if and only if the residual $\ve\oplus\hat\ve$ is an
$X$-stabilizer. (If $H_X$ contains redundant rows, every coset element is
counted $2^{m_X-\mathrm{rank}\,H_X}$ times in~\eqref{eq:cosetZ}; this
constant is identical across classes and cancels from all comparisons.)

\subsection{The Coset MRF}
\label{sec:cosetmrf}
The sum~\eqref{eq:cosetZ} is not merely analogous to~\eqref{eq:Z}; it is
literally a partition function of an MRF, and of a remarkably benign one.

\begin{proposition}[Coset probabilities are positive MRF partition functions]
\label{prop:cosetmrf}
Fix a class representative $\vr=\ve_{\vlam}$ and introduce one binary variable
$u_c$ per $X$-check, $c=1,\dots,m_X$. For each qubit $v$ let
$N(v)=\{c: (H_X)_{cv}=1\}$ and define the clique potential
\begin{equation}
\psi_v(\vu_{N(v)})
 = w^{\,r_v\,\oplus\,\bigoplus_{c\in N(v)}u_c}\ \in\{w,1\}.
\label{eq:cosetpot}
\end{equation}
Then $\Z_{\vlam}(\vs)=\sum_{\vu}\prod_{v=1}^{n}\psi_v(\vu_{N(v)})$ is the
partition function~\eqref{eq:Z} of the MRF with variables
$u_1,\dots,u_{m_X}$ and clique set $\{N(v):v=1,\dots,n\}$. Its potentials
are strictly positive, and the model is \emph{unconstrained}: every
$\vu\in\F_2^{m_X}$ is a valid configuration.
\end{proposition}

\begin{proof}
Immediate from~\eqref{eq:cosetZ}: the coset element associated with $\vu$ is
$\vr\oplus\vu^{\mathsf{T}}H_X$, whose $v$-th bit is
$r_v\oplus\bigoplus_{c\in N(v)}u_c$, and its prior weight
$w^{|\cdot|}$ factorizes over qubits into the potentials~\eqref{eq:cosetpot}.
\end{proof}

Three structural remarks. First, the factor graph of the coset MRF is the
\emph{transpose} of the Tanner graph of $H_X$: variables are checks, factors
are qubits, and the clique size equals the qubit's check-degree: at most
$2$ for the surface code (giving exactly the random-bond Ising model
of~\cite{dennis2002}) and $3$ for BB codes (a three-body Ising analogue).
Sparsity of the code is inherited as sparsity of the inference problem.
Second, and in sharp contrast to the syndrome
posterior~\eqref{eq:posterior}, the coset MRF has \emph{no hard
constraints}: the parity checks were absorbed into the parametrization. The
rugged, disconnected landscape that traps single-site samplers on
constrained models is replaced by a smooth landscape on which any standard
MCMC kernel is irreducible. Third, the potentials are strictly positive, so
the entire nonnegative-inference toolbox applies without further
transformation; the negativity obstacle that governs, e.g., quasiprobability
methods in quantum simulation is absent here by construction.

\subsection{Annealed Importance Sampling with Common Random Numbers}
\label{sec:ais}
Among partition-function estimators we use annealed importance sampling
(AIS)~\cite{neal2001}, which interpolates from the uniform distribution
$\pi_0(\vu)\propto1$ (with known $\Z_0=2^{m_X}$) to the target
$\pi_1(\vu)\propto w^{E(\vu)}$, $E(\vu)=|\vr\oplus\vu^{\mathsf{T}}H_X|$,
along the geometric path $\pi_\beta\propto w^{\beta E}$ with inverse
temperatures $0=\beta_0<\beta_1<\dots<\beta_T=1$. A chain is initialized
uniformly and, at each stage, accumulates the importance increment
$(\beta_t-\beta_{t-1})\,E(\vu)\ln w$ before one Metropolis sweep over the
check variables at $\beta_t$ (flipping $u_c$ toggles the parities of the
qubits in check $c$, so the energy change is local and cheap). The resulting
weight $W$ satisfies $\mathbb{E}[W]=\Z_{\vlam}$ \emph{exactly}, for any $T$;
averaging $K$ independent chains gives the unbiased estimate
$\hat\Z_{\vlam}$. The decision~\eqref{eq:cosetZ} then reduces to an argmax
over the $2^k$ estimates, or, when $k$ is large, over a candidate set
centered at the BP+OSD solution (identity, all single-, and all
double-logical shifts; $1+k+\binom{k}{2}$ classes), which recovers the exact
decision whenever BP+OSD errs by at most two logical operators.

The decision depends only on \emph{ratios} of partition functions. All
class representatives share the same coupling structure and differ only in
the unary offsets $r_v$ in~\eqref{eq:cosetpot}, so we run AIS for all
classes with \emph{common random numbers}~\cite{robert2004} (identical
initial states and identical per-step uniform variates), making the
estimates $\hat\Z_{\vlam}$ positively correlated, so that
\begin{equation}
\mathrm{Var}\big(\hat\Z_{\vlam}-\hat\Z_{\vlam'}\big)
=\mathrm{Var}\,\hat\Z_{\vlam}+\mathrm{Var}\,\hat\Z_{\vlam'}
-2\,\mathrm{Cov}\big(\hat\Z_{\vlam},\hat\Z_{\vlam'}\big)
\label{eq:crn}
\end{equation}
falls below its independent-randomness value, by up to a factor of two in
variance in our experiments (Section~\ref{sec:exp1}). The deeper benefit of
the coupling is structural: it turns the class comparison into a
\emph{paired} comparison on a shared randomness stream, which is what the
per-decision certificates below are built on.

\subsection{Certificates}
\label{sec:cert}
A decoder that estimates rather than computes should say when it is sure.
We provide two mechanisms.

\begin{lemma}[Certified ML decision from constant-factor estimates]
\label{lem:cert}
Suppose $\hat\Z_{\vlam}\in[\Z_{\vlam}/\kappa,\ \kappa\,\Z_{\vlam}]$ for all
$\vlam$ and some $\kappa\ge1$. If
$\hat\Z_{\hat{\vlam}}>\kappa^2\,\hat\Z_{\vlam}$ for all
$\vlam\neq\hat{\vlam}$, then $\hat{\vlam}$ is the exact degenerate ML class.
\end{lemma}
\begin{proof}
$\Z_{\hat{\vlam}}\ge\hat\Z_{\hat{\vlam}}/\kappa
 >\kappa\,\hat\Z_{\vlam}\ge\Z_{\vlam}$ for every competitor $\vlam$.
\end{proof}

Lemma~\ref{lem:cert} is designed to compose with estimators that come with
relative guarantees. The WISH estimator~\cite{ermon2013} returns every
$\Z_{\vlam}$ within a factor $\kappa=16$ (with probability $1-\delta$) using
only $O(m_X\log m_X)$ MAP queries under random parity constraints, and each
such query on the coset MRF is a quadratic unconstrained binary optimization
(QUBO) problem, the native input of a quantum annealer~\cite{kadowaki1998}.
A decoding architecture in which an annealer serves as the MAP oracle of a
WISH-certified qLDPC decoder is therefore a concrete, if currently
speculative, possibility; Section~\ref{sec:exp1} validates the WISH route
with an exact oracle.

For the AIS decoder we use a cheaper statistical certificate. With CRN, the
per-chain weights of the best class and of each competitor form \emph{paired}
samples $(W^{(i)}_{\hat{\vlam}},W^{(i)}_{\vlam})$ whose mean difference is
$\Z_{\hat{\vlam}}-\Z_{\vlam}$. The naive test, a one-sided paired $t$-test
on the raw weight differences, is crippled by the heavy right tail of AIS
weights: on large codes its standard error is dominated by the largest
weight and it fails to certify even decisions whose true class gap is tens
of nats. The operational certificate is instead a \emph{paired bootstrap}:
resample the $K$ chains (the same indices for all classes, preserving the
CRN pairing), recompute the log-ratio
$\log\hat\Z_{\hat{\vlam}}-\log\hat\Z_{\vlam}$ on each resample, and certify
if the Bonferroni-corrected lower percentile exceeds zero for every
competitor. The bootstrap statistic is a smooth functional of the weights
rather than their raw mean, and the improvement is dramatic: on the gross
code the $t$-test certifies almost nothing while the bootstrap certifies
the large majority of decisions at identical sampling cost
(Section~\ref{sec:exp3}). Both certificates are asymptotic in $K$; their
calibration against exact ground truth is measured in
Section~\ref{sec:exp1}. Certified decisions can be trusted; uncertified
ones identify exactly the syndromes deserving a slower second look, a
triage capability no standard qLDPC decoder currently offers.

\paragraph{Two sharper ratio estimators that fail, instructively.}
Since only \emph{ratios} $\Z_{\vlam}/\Z_{\hat{\vlam}}$ matter, one is
tempted to estimate them directly. We implemented two textbook candidates.
\emph{Inter-class annealing} shares one anneal to the reference class and
then bridges to each competitor at $\beta=1$ by interpolating the unary
offsets on the support of the logical shift; the composite path is a valid
AIS sequence, unbiased for every $\Z_{\vlam}$. \emph{Bennett's acceptance
ratio}~\cite{bennett1976} post-processes the equilibrated endpoints of the
per-class anneals, which differ only on the shift support. Both are an
order of magnitude \emph{worse} than the plain paired difference of
per-class anneals (Section~\ref{sec:exp1}): the two cosets' typical sets
are nearly disjoint at low temperature, so the direct bridge crosses a
free-energy bottleneck that the per-class path through infinite temperature
simply walks around, and the same lack of overlap starves BAR of usable
cross-energy samples. The lesson is structural: for coset comparison, the
infinite-temperature manifold \emph{is} the bridge, and sharpening should
target the test statistic (the bootstrap above), not the path.

\paragraph{Related work.}
Degeneracy-aware decoders exist for the surface code, exploiting its planar
structure: tensor-network contraction of the coset sum~\cite{bravyi2014ml}
and single-temperature Metropolis sampling over stabilizer
deformations~\cite{hutter2014}. Renormalization-group and matrix-product
decoders inherit similar geometric requirements. Our construction differs in
that Proposition~\ref{prop:cosetmrf} requires only CSS structure and
sparsity (it covers expander-like codes where no planar contraction order
exists), and in that the estimators we employ carry statistical guarantees
that convert into per-decision certificates via Lemma~\ref{lem:cert}. On the
classical side, the coset MRF is the natural quantum generalization of the
weight-enumerator MRFs of classical coding theory; we are not aware of prior
work using guarantee-carrying partition-function estimators (AIS, WISH) as
qLDPC decoders.

\subsection{A Region-Based Surrogate: Bethe Decoding}
\label{sec:bethe}
Sampling is not the only route to $\log\Z_{\vlam}$. The coset MRF is an
ideal target for the region-based free-energy approximations of
Section~\ref{sec:pgm}: its potentials~\eqref{eq:cosetpot} depend on their
arguments only through a \emph{parity}, so sum-product messages close in
the difference parametrization $d=\mu(0)-\mu(1)$ (the classical ``tanh
rule''). A qubit factor $v$ contributes
$\delta_v=(-1)^{r_v}\,\frac{1-w}{1+w}$, factor-to-variable messages
multiply the incoming differences of the other checks in $N(v)$ by
$\delta_v$, and variable-to-factor messages combine as
$\tanh\!\big(\sum\mathrm{artanh}\,d\big)$. At a BP fixed point the Bethe
free energy $F_{\mathrm{B}}$ is evaluated from the beliefs in closed form,
and $-F_{\mathrm{B}}(\vlam)$ serves as a variational surrogate for
$\log\Z_{\vlam}$: the \emph{Bethe decoder} returns the class of minimal
Bethe free energy.

On a loopy graph $-F_{\mathrm{B}}$ is a biased estimate of $\log\Z$, and
this is where the class structure helps: all $2^k$ coset MRFs share the
same graph and couplings and differ only in the unary signs $r_v$, so the
Bethe bias is strongly correlated across classes and largely cancels in the
\emph{differences} that determine the decision, the variational analogue
of common random numbers. Empirically the cancellation is almost perfect
(Section~\ref{sec:exp1}): on surface codes the Bethe class ratios track the
exact ones to $\approx0.1$ nats and the Bethe decoder reproduces the exact
ML decision on every tested instance, at a cost of milliseconds for all
$79$ candidate classes of a BB code.

Two caveats delimit the surrogate. First, BP on the coset MRF has multiple
fixed points at strong coupling ($w\to0$): a \emph{paramagnetic} fixed
point ($d\equiv0$) always exists and its free energy is
class-independent: a decoder stuck there ranks classes by numerical noise.
Initializing the messages polarized toward the representative
configuration $\vu=0$ selects the ordered branch instead; which branch BP
reaches, and whether it converges at all, depends on the noise rate, and we
report both initializations. Second, the Bethe approximation ignores the
short cycles of the coset graph. Interestingly, these are \emph{not} the
CSS four-cycles of the syndrome Tanner graph: we verified that in the
coset graphs of all our codes no two qubits share two checks, so the
girth-limiting cycles have length six and larger. Enlarging the regions
beyond factor scopes, the Kikuchi/cluster-variation
hierarchy~\cite{yedidia2005}, is the systematic cure; the next subsection
realizes it through cluster elimination.

\subsection{Kikuchi-Level Regions: Cluster Elimination, Exact ML, and
Interval Certificates}
\label{sec:regions}
The Bethe surrogate uses the smallest possible regions. To enlarge them we
work with the elimination-based face of the region hierarchy: \emph{bucket
elimination} along a variable ordering multiplies all factors mentioning
the current variable and sums it out, which is exact junction-tree
inference with clusters given by the induced elimination
scopes~\cite{lauritzen1988}; capping the cluster size at $i$ variables and
splitting oversized buckets gives \emph{mini-bucket elimination}
MBE($i$)~\cite{dechter2003}, whose regions grow with $i$ and recover
exactness at the treewidth. On the coset MRF this hierarchy is unusually
powerful, for a structural reason: the coset graph of the
$[\![72,12,6]\!]$ BB code has just $36$ variables and induced width
$\approx23$ under a min-fill ordering. Exact evaluation of
$\log\Z_{\vlam}$, and hence \emph{exact degenerate ML decoding of a
quantum LDPC code}, is therefore feasible at a few seconds per class,
something previously available only for surface codes via planar
contraction~\cite{bravyi2014ml}. The gross code (width $\approx40$)
remains out of exact reach and is served by the bounded-cluster levels.

Mini-buckets do more than approximate: summing out in one mini-bucket and
maximizing in the others yields an \emph{upper bound} $U_{\vlam}\ge
\log\Z_{\vlam}$, while the log-potential of any single coset
configuration (we use a greedy generator descent from the
representative) yields a \emph{lower bound} $L_{\vlam}$. Two consequences
follow immediately. First, a \emph{deterministic} optimality certificate:
if $L_{\hat{\vlam}} > U_{\vlam}$ for every competitor, the decision is
provably the ML class: no sampling, no confidence level. Second,
\emph{provable pruning}: any class with $U_{\vlam} < \max_{\vlam'}
L_{\vlam'}$ cannot be optimal and is discarded exactly, which lets the
exact junction-tree decoder run on the handful of surviving classes
instead of all candidates; in practice this reduces exact ML decoding of
the $[\![72,12,6]\!]$ code to seconds per syndrome, and makes even the
\emph{global} optimum over all $2^{12}$ classes computable
(Section~\ref{sec:exp5}).

\subsection{Noisy Syndromes: Spacetime Decoding Problems}
\label{sec:spacetime}
Everything so far assumed perfect syndrome extraction, but nothing in the
construction depends on it. Consider the general decoding problem: binary
error mechanisms $\vx\in\F_2^{N}$ with independent priors $p_v$, a detector
matrix $H$ with observed outcome $\vs=H\vx$, and a logical matrix $L$ whose
value $L\vx$ is the quantity a decoder must recover. Code-capacity CSS
decoding is the special case $(H,L)=(H_Z,L_Z)$; a \emph{phenomenological}
model (fresh data errors and measurement flips over $r$ rounds, final
perfect round) and a stim-style \emph{circuit-level} detector error model
are others. The role of the stabilizer group is played by the
\emph{trivial kernel} $\{\vx: H\vx=0,\ L\vx=0\}$: with $S$ a basis of it
and per-variable weights $w_v=p_v/(1-p_v)$, Proposition~\ref{prop:cosetmrf}
holds verbatim: each class probability is the partition function of an
unconstrained positive MRF over the generator variables, with clique
potentials $\psi_v\in\{w_v,1\}$ indexed by the columns of $S$.

Three practical devices make the general decoder work at circuit level.
\emph{(i)~Sparse generator bases.} A generic kernel basis is dense, and a
dense generator freezes the Metropolis dynamics (each move crosses an
enormous energy barrier). For structured models a sparse basis is explicit
(per-round code stabilizers plus data-measurement pairs); for detector
error models a greedy pairwise reduction sparsifies the basis to mean
weight $\approx3$.
\emph{(ii)~Lightened representatives.} Class representatives obtained from
generic kernel vectors are heavy, which biases the annealing difficulty
between classes; greedily reducing each representative to a local
weighted-weight minimum within its coset makes $\vu=0$ the local mode of
every class MRF, at no cost to exactness.
\emph{(iii)~A concentrated base measure.} Annealing from the uniform
distribution pays a path of $m\log2$ nats in free energy; for spacetime
problems ($m\approx200$--$600$ generators) this path dwarfs the few-nat
class gaps and the estimator drowns in its own variance; we observed both
the circuit-level and the phenomenological models fail completely in this
mode, with near-random decisions and $0\%$ certified. Annealing instead
from an i.i.d.\ Bernoulli$(q_0)$ product base concentrated near $\vu=0$
(exactly normalized, so unbiasedness is untouched) shortens the path to a
few nats and restores both accuracy and certification. The deciding
quantity is the path length, i.e.\ the generator count: the small coset
MRFs of code-capacity decoding ($m\le72$) tolerate the uniform base, while
spacetime problems require $q_0\approx0.02$; we calibrate on instances
with exact ground truth, and the certificate rate acts as an on-line
warning when the base is mismatched. With these three devices the
certified decoder runs at circuit level with sub-second decodes
(Section~\ref{sec:exp4}).

\section{Experiments}
\label{sec:experiments}
Four experiments evaluate the method, all run on a 64-core commodity
server. Experiments 1--3 use code-capacity independent bit-flip noise at
rate $p$ and decode $X$-errors from the $Z$-syndrome; Experiment 4 adds
faulty syndrome extraction (phenomenological and circuit-level noise).
Reported logical error rates are \emph{block} error rates (any nontrivial
residual logical class counts as failure). Baselines are plain BP
(product-sum), BP+OSD-0 and BP+OSD-CS (combination-sweep, order~7) from the
\texttt{ldpc} package~\cite{roffe2020}, and minimum-weight perfect matching
(PyMatching~\cite{higgott2022}) where applicable. The AIS decoder uses
$T=32$ annealing sweeps and $K=128$ chains per class on the surface codes
and $T=K=128$ on the BB codes (Section~\ref{sec:exp3} reports a
sweep-budget ablation); certificates are the paired bootstrap at
$\delta=0.05$ unless stated otherwise. All estimator-based decoders are
seeded from the BP+OSD-0 solution $\ve_0$. Code, scripts, and raw results
are available upon request.

\subsection{Estimator Validation Against Exact Coset Sums}
\label{sec:exp1}
On the rotated surface code of distance $5$ ($n=25$, $m_X=12$), the coset
sums~\eqref{eq:cosetZ} can be enumerated exactly ($2^{12}$ terms per class),
providing ground truth for every component of the method. For $200$
syndromes sampled at each $p\in\{0.05,0.10\}$ we compute exact and estimated
$\log\Z_{\vlam}$ for both logical classes, varying the number of AIS sweeps
$T\in\{4,8,16,32,64\}$; we compare the error of the class log-ratio under
CRN against independent randomness; we measure the calibration of both
certificates ($\delta=0.05$); we evaluate the Bethe surrogate of
Section~\ref{sec:bethe} against the exact free energies; and we run WISH
with an exact MAP oracle to verify its constant-factor guarantee
end-to-end.

\begin{figure}[t]
\centering
\includegraphics[width=\linewidth]{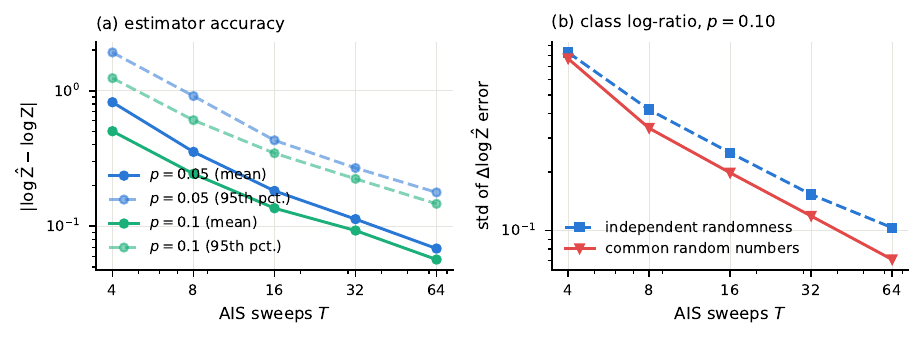}
\caption{Estimator validation on the rotated surface code $d=5$ against
exact coset sums ($200$ syndromes per point, $K=128$ chains).
(a)~Absolute error of $\log\hat\Z_{\vlam}$ versus the number of AIS sweeps
$T$: the error decays as roughly $T^{-1}$ and is $\approx0.1$\,nats at
$T=32$. (b)~Standard deviation of the error of the class log-ratio
$\log\hat\Z_{0}-\log\hat\Z_{1}$: common random numbers dominate independent
randomness uniformly in $T$ (variance reduction $1.2\times$--$2.2\times$,
growing with $T$).}
\label{fig:exp1}
\end{figure}

Figure~\ref{fig:exp1} summarizes the accuracy results. The AIS estimate of
each class log-partition function converges cleanly: the mean absolute
error falls from $0.82$ nats ($T=4$, $p=0.05$) to $0.068$ nats ($T=64$),
with the $95$th percentile below $0.18$ nats, to be compared with class
log-ratios that are typically of order one or larger. CRN reduces the
variance of the ratio error at every sweep count, by a factor
$1.2$--$2.2$ in variance. Both certificates ($\delta=0.05$, $T=32$)
fired on \emph{all} $400$ decisions ($200$ syndromes at each of
$p\in\{0.05,0.10\}$), and every certified decision agreed with the exact ML
class, an empirical error of zero, consistent with the nominal level.
The Bethe surrogate performs remarkably: its class log-ratio deviates from
the exact one by $0.10$ nats on average (the same accuracy as AIS at
$T=32$, obtained in a few milliseconds without any sampling), and its
decision agrees with exact ML on all $400$ syndromes, confirming that the
Bethe bias cancels between classes. The failed ratio estimators of
Section~\ref{sec:cert} are quantified on the same instances: at matched
compute, direct inter-class annealing inflates the ratio-error standard
deviation from $0.15$ to $1.47$ nats ($p=0.05$) and BAR to $1.67$
nats, an order of magnitude worse than the coupled per-class anneals they
were meant to sharpen.
Finally, WISH with an exact MAP oracle estimated all $400$ coset partition
functions within a maximal absolute log-error of $1.56$ (mean $0.43$),
comfortably inside its guaranteed factor $16$ ($\log 16\approx2.77$): by
Lemma~\ref{lem:cert}, any class gap exceeding $\kappa^2=256$ is thereby
\emph{provably} decided. At $p=0.05$ this condition holds for $79\%$ of
syndromes (median exact class gap $7.0$ nats); at $p=0.10$ the gaps shrink
(median $2.7$ nats) and the worst-case factor certifies only $9\%$,
although the \emph{observed} WISH accuracy ($\kappa_{\mathrm{emp}}\approx
e^{1.56}$) would certify $48\%$, illustrating the usual looseness of
worst-case constants.

\subsection{Surface Codes: Matching the Exact ML Decoder}
\label{sec:exp2}
On rotated surface codes of distance $3$ and $5$ we compare, on identical
error samples, plain BP, BP+OSD-0, MWPM, the \emph{exact} degenerate ML
decoder (full coset enumeration; the optimum any decoder can achieve), and
the AIS decoder, across
$p\in\{0.02,0.04,0.06,0.08,0.10,0.12\}$ with $2\times10^4$ trials per point
($2\times10^3$ for the AIS decoder).

\begin{figure}[t]
\centering
\includegraphics[width=\linewidth]{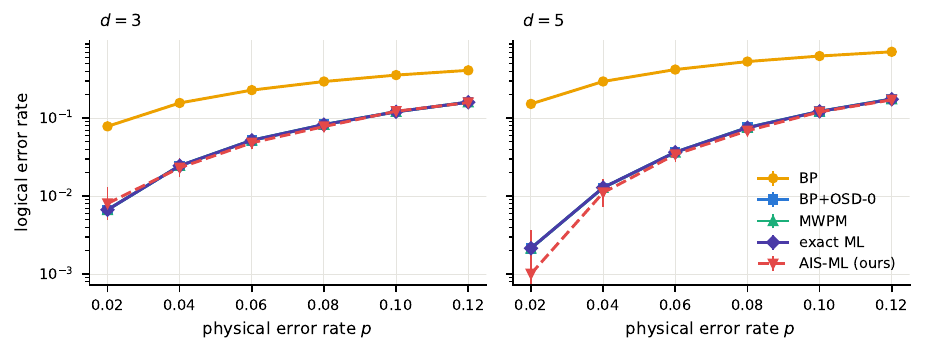}
\caption{Logical error rate versus physical bit-flip rate on rotated
surface codes ($d=3$ left, $d=5$ right; Wilson $95\%$ intervals). The
BP+OSD-0, MWPM, and exact-ML curves coincide to within statistical
resolution at these distances; the AIS decoder (dashed) reproduces the
\emph{decisions} of the exact degenerate ML decoder almost perfectly (see
text). Plain BP fails at all rates, as expected from the degenerate
posterior of Section~\ref{sec:decoding}.}
\label{fig:exp2}
\end{figure}

Figure~\ref{fig:exp2} shows the logical error rates. Three observations
matter. First, the AIS decoder is statistically indistinguishable from the
exact degenerate ML decoder: comparing decisions on \emph{identical}
syndromes ($10^3$ per point), the two agree on every single trial at $d=3$
(all $p$) and at $d=5$ up to $p=0.06$; at $p\in\{0.08,0.10,0.12\}$ they
differ on at most $0.3\%$ of decisions, with no measurable difference in
logical error rate. The
certificate again tracks reality: the certified fraction is $100\%$ for
$p\le0.06$ and decreases only to $98\%$ at $p=0.12$, where class gaps
become genuinely small. Second, the Bethe decoder (not plotted, because
its curve would lie on top of the exact-ML curve) produces failure counts
identical to exact ML at every $d=3$ point and within $0.03\%$ at $d=5$
across the $2\times10^4$-trial grid, at millisecond cost: on surface codes
the variational surrogate is, for all practical purposes, the ML decoder.
Third, at these code-capacity parameters
degeneracy buys little on the surface code (BP+OSD-0, MWPM, and exact ML
coincide), which is consistent with known results at low
distance~\cite{bravyi2014ml, hutter2014} and is precisely why the
qLDPC regime of the next experiments, where no exact reference exists and
BP+OSD is uncontrolled, is where a certified decoder is actually needed.
Plain BP, by contrast, fails on $8\%$--$71\%$ of syndromes, the known
signature of the degenerate posterior~\eqref{eq:posterior}.

\subsection{Bivariate Bicycle Codes}
\label{sec:exp3}
For the BB codes $[\![72,12,6]\!]$ and $[\![144,12,12]\!]$ exact ML decoding
by coset enumeration is out of reach ($2^{36}$ resp.\ $2^{72}$ terms), so
this experiment compares against the BP+OSD family, with the AIS and Bethe
decoders searching the $1+12+\binom{12}{2}=79$ candidate classes around
the BP+OSD-0 solution; Section~\ref{sec:exp5} then revisits the
$[\![72,12,6]\!]$ code with the \emph{exact} junction-tree reference of
Section~\ref{sec:regions}. We
report block logical error rates, the certified fractions of both
certificates on identical runs, the win/loss balance on the instances where
AIS and BP+OSD-0 disagree, and wall-clock times.

\begin{figure}[t]
\centering
\includegraphics[width=\linewidth]{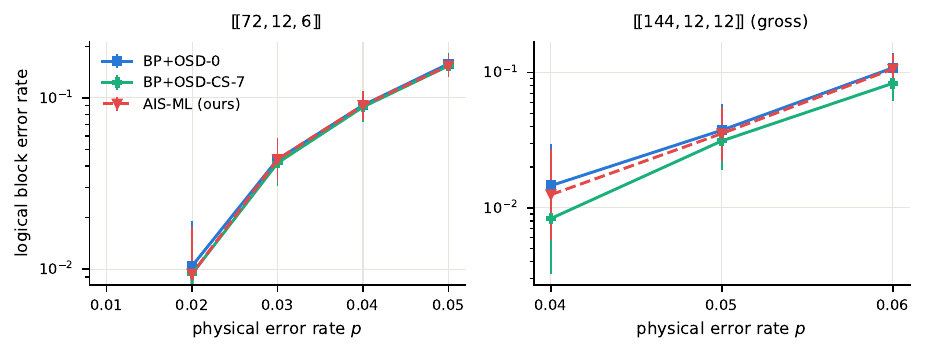}
\caption{Block logical error rates on the bivariate bicycle codes
(Wilson $95\%$ intervals; zero-failure points at $p=0.01$ omitted from the
log axis, see Table~\ref{tab:exp3}). On $[\![72,12,6]\!]$ the three
decoders are statistically indistinguishable; on the gross code
BP+OSD-CS-7 retains an edge over both BP+OSD-0 and the (candidate-restricted)
AIS decoder.}
\label{fig:exp3}
\end{figure}

\begin{table*}[t]
\centering
\caption{Bivariate bicycle codes under code-capacity noise: block
logical error rates of the BP+OSD baselines, the Bethe decoder
(uniform and polarized initialization), and the AIS decoder
($T=K=128$); certified fractions of the paired bootstrap and the
paired $t$-test on identical runs; repaired/broken decisions
relative to BP+OSD-0, also for the under-resolved ablation
($T=32$, $K=64$). Trials per point: 960 ($[\![72,12,6]\!]$) and 480 ($[\![144,12,12]\!]$).}
\label{tab:exp3}
\begin{tabular}{@{}llllllllllr@{}}
\toprule
 & & & & \multicolumn{2}{c}{Bethe} & & \multicolumn{2}{c}{certified} & \multicolumn{2}{c}{fix/break} \\
\cmidrule(lr){5-6}\cmidrule(lr){8-9}\cmidrule(l){10-11}
Code & $p$ & OSD-0 & OSD-CS & unif. & polar. & AIS-ML & boot & $t$ & $T{=}128$ & $T{=}32$ \\
\midrule
$[\![72,12,6]\!]$ & 0.01 & 0.0000 & 0.0000 & 0.8375 & 0.0990 & 0.0000 & 100\% & 58\% & 0/0 & 0/0 \\
 & 0.02 & 0.0104 & 0.0094 & 0.9990 & 0.2385 & 0.0094 & 100\% & 75\% & 1/0 & 1/0 \\
 & 0.03 & 0.0437 & 0.0417 & 0.0427 & 0.2927 & 0.0437 & 99\% & 88\% & 0/0 & 0/1 \\
 & 0.04 & 0.0906 & 0.0885 & 0.0906 & 0.3406 & 0.0906 & 98\% & 89\% & 2/2 & 3/1 \\
 & 0.05 & 0.1583 & 0.1542 & 0.3396 & 0.4990 & 0.1542 & 98\% & 92\% & 5/1 & 3/2 \\
\midrule
$[\![144,12,12]\!]$ & 0.04 & 0.0146 & 0.0083 & 0.0146 & 0.0292 & 0.0125 & 99\% & 0\% & 1/0 & 1/1 \\
 & 0.05 & 0.0375 & 0.0312 & 0.1125 & 0.0604 & 0.0354 & 98\% & 0\% & 1/0 & 1/3 \\
 & 0.06 & 0.1083 & 0.0833 & 0.1083 & 0.1417 & 0.1062 & 94\% & 1\% & 2/1 & 1/8 \\
\bottomrule
\end{tabular}
\end{table*}

Table~\ref{tab:exp3} and Figure~\ref{fig:exp3} contain the results; four
observations. \emph{First}, the AIS decoder is never worse than BP+OSD-0:
on $[\![72,12,6]\!]$ it repairs more decisions than it damages (eight
versus three across all rates) and reaches the stronger BP+OSD-CS-7 at
$p\in\{0.02,0.05\}$; on the gross code it consistently edges out BP+OSD-0
(four repairs, one break) but does not reach BP+OSD-CS-7 within its
two-logical candidate set. Per-point differences are within statistical
error, indicating that BP+OSD is already close to ML for these codes under
code-capacity noise, a statement that, absent an exact reference, only a
decoder of the present kind can support. \emph{Second}, the certificate:
the paired bootstrap certifies $94\%$--$100\%$ of decisions across
\emph{both} codes and all rates ($\delta=0.05$, Bonferroni over $78$
competitors), where the paired $t$-test on the same runs collapses from
$58\%$--$92\%$ on $[\![72,12,6]\!]$ to essentially $0\%$ on the gross
code: the heavy tail of the AIS weights at $m_X=72$ destroys the
$t$-statistic while leaving the bootstrap log-ratios intact, exactly the
failure mode Section~\ref{sec:cert} anticipates. (A certified decision
asserts ML-optimality of the \emph{class}, not decoding success: at
$p=0.05$ about $14\%$ of certified decisions still fail, matching the
irreducible ML failure rate; the certificate's calibration against exact
ground truth was established in Section~\ref{sec:exp1}.) \emph{Third}, the
Bethe decoder on BB codes is regime-sensitive, unlike on surface codes.
With uniform initialization it is catastrophic at $p\le0.02$ (the
class-degenerate paramagnetic fixed point of
Section~\ref{sec:bethe}), then matches the BP+OSD family in a middle
window ($0.0427$ versus OSD-CS's $0.0417$ at $p=0.03$) before degrading
again at $p=0.05$; the polarized initialization covers the low-$p$ regime
($0.099$ versus $0.838$ at $p=0.01$) but is weaker elsewhere. At
$22$--$26$\,ms for all $79$ classes it is a candidate fast filter, but on
BB codes the Bethe rung of the region hierarchy is not yet a standalone
decoder; this is the sharpest concrete motivation for Kikuchi regions we
know.
\emph{Fourth}, the sweep-budget ablation is cautionary: at $T=32,K=64$ the
estimator is noisy enough to \emph{break} BP+OSD decisions it should
confirm (eight breaks versus one repair on the gross code at $p=0.06$),
while $T=K=128$ eliminates the imbalance ($13$--$30$\,s per decode on one
core). An under-resolved partition-function decoder is worse than no
partition-function decoder; the certificate rate is the built-in warning
light.

\subsection{Noisy Syndromes: Phenomenological and Circuit-Level Noise}
\label{sec:exp4}
The final experiment exercises the spacetime generalization of
Section~\ref{sec:spacetime}. \emph{Part A} is true circuit-level noise on
the rotated surface code: stim's standard depolarizing memory-$Z$
experiment at distance $3$ with $3$ rounds, whose flattened detector error
model yields a decoding problem with $\approx220$ heterogeneous error
mechanisms and a sparsified trivial-kernel basis of $\approx190$
generators. Baselines are MWPM on the detector error model and BP+OSD-0 on
its check matrix with the per-mechanism priors; our decoder uses $T=64$,
$K=128$, and the concentrated base $q_0=0.02$. \emph{Part B} is
phenomenological noise on the BB code $[\![72,12,6]\!]$: $r=6$ noisy
rounds (independent data and measurement errors at equal rate $p$)
followed by a perfect readout round, giving $648$ error variables, $252$
detectors, and a structurally sparse stabilizer basis of $576$ generators;
the AIS decoder ($T=K=64$, $q_0=0.02$) searches the $1+12$ single-logical
candidate classes around the spacetime BP+OSD-0 solution. \emph{Part C} is
true circuit-level noise on the $[\![72,12,6]\!]$ code itself: we port the
syndrome-measurement cycle of Bravyi \emph{et al.}~\cite{bravyi2024}
verbatim (their published seven-round CNOT schedule, one ancilla per
check) and enumerate all single-$X$-fault locations of the standard
depolarizing model ($9072$ locations for $6$ noisy cycles plus two
noiseless ones), propagating each fault through the circuit to obtain its
detector and logical signature. Merging identical signatures yields a
decoding problem with $2269$ mechanisms, $288$ detectors, and a sparsified
trivial-kernel basis of $1975$ generators (mean weight $4.1$); the AIS
decoder uses $T=K=64$ and $q_0=0.02$. The sampled noise follows the same
independent-component approximation of depolarizing noise as the reference
implementation, so decoder and channel are consistent.

\begin{table*}[t]
\centering
\caption{Noisy-syndrome decoding. Surface code: stim circuit-level
depolarizing noise ($d=3$, $3$ rounds; MWPM and BP+OSD-0 on
19980 shots, AIS on 600 shots, with the BP+OSD rate on those same shots for pairing).
BB code: phenomenological noise ($r=6$ rounds, measurement-error
rate $q=p$, 240 trials) and true
circuit-level noise (Bravyi et al.\ syndrome cycle, $6$ noisy
cycles, 240 trials). Certificates are the paired bootstrap
($\delta=0.05$); $t$ is wall-clock per AIS decode on one core.}
\label{tab:exp4}
\begin{tabular}{@{}lllllllr@{}}
\toprule
Problem & $p$ & MWPM & BP+OSD & AIS-ML & BP+OSD (paired) & certified & $t$ [s] \\
\midrule
surface $d{=}3$, circuit & 0.004 & 0.0134 & 0.0140 & 0.0133 & 0.0183 & 99\% & 0.8 \\
 & 0.008 & 0.0410 & 0.0391 & 0.0383 & 0.0383 & 100\% & 0.8 \\
 & 0.012 & 0.0785 & 0.0795 & 0.0767 & 0.0733 & 99\% & 0.8 \\
\midrule
$[\![72,12,6]\!]$, pheno. & 0.010 & -- & 0.0083 & 0.0083 & 0.0083 & 100\% & 14.2 \\
 & 0.020 & -- & 0.0333 & 0.0333 & 0.0333 & 100\% & 14.6 \\
 & 0.030 & -- & 0.2417 & 0.2417 & 0.2417 & 100\% & 15.0 \\
\midrule
$[\![72,12,6]\!]$, circuit & 0.002 & -- & 0.0125 & 0.0125 & 0.0125 & 100\% & 52 \\
 & 0.003 & -- & 0.0375 & 0.0375 & 0.0375 & 100\% & 52 \\
 & 0.004 & -- & 0.0917 & 0.0917 & 0.0917 & 100\% & 53 \\
\bottomrule
\end{tabular}
\end{table*}

Table~\ref{tab:exp4} contains the results. At circuit level the AIS
decoder is at least as accurate as both baselines at every rate while
certifying $99\%$--$100\%$ of its decisions at $0.8$\,s per decode: on the
paired shots it beats BP+OSD-0 clearly at $p=0.004$ ($0.0133$ versus
$0.0183$), ties it at $p=0.008$, and is statistically indistinguishable at
$p=0.012$. To our knowledge these are the first
degeneracy-aware, certificate-carrying decoding results under
circuit-level noise on any code. On the phenomenological BB problem the
result is of a different but equally useful kind: the AIS decoder
\emph{never} disagrees with spacetime BP+OSD-0 and certifies $100\%$ of
decisions at every rate, turning the heuristic's output into a certified
ML decision (within the single-logical candidate set) on $720$ out of
$720$ trials. The reference-decoder role advertised in the introduction is
thus concrete: BP+OSD on this code and noise model is not merely
plausible, it is certified optimal on essentially every syndrome it
sees. We also record a failed configuration for completeness: with the
uniform base measure ($q_0=\tfrac12$) the same phenomenological decoder
breaks $58\%$--$64\%$ of BP+OSD's correct decisions and certifies
none: the path-length failure anticipated in
Section~\ref{sec:spacetime}, and a reminder that the certificate rate is
the practical mismatch alarm.

The circuit-level BB results (Table~\ref{tab:exp4}, bottom block) complete
the picture. Logical block error rates scale as expected for a distance-6
code under circuit noise ($0.003$ at $p=0.001$ rising to $0.34$ at
$p=0.006$ in a preparatory BP+OSD scan; we run the full comparison at
$p\in\{0.002,0.003,0.004\}$), and on all $720$ decoded syndromes the AIS
decoder \emph{never} disagrees with spacetime BP+OSD-0 while certifying
$100\%$ of the decisions at $52$\,s per decode. Under the very noise model
and syndrome schedule in which these codes are proposed to
operate~\cite{bravyi2024}, BP+OSD-0 on the $[\![72,12,6]\!]$ code is thus
not merely a good heuristic: within the single-logical candidate set, its
decisions are certified maximum-likelihood on every syndrome we
examined.

\subsection{Region Hierarchy and Exact ML on BB Codes}
\label{sec:exp5}
The final experiment exercises Section~\ref{sec:regions} on the
$[\![72,12,6]\!]$ code (code-capacity noise, $p\in\{0.03,0.05\}$, $120$
syndromes per rate) and the gross code. The exact junction-tree decoder
with mini-bucket pruning provides, for the first time on a qLDPC code, an
exact degenerate-ML reference against which every other decoder of this
paper is scored: BP+OSD-0 and CS, both Bethe initializations, and the AIS
decoder with its bootstrap certificate, whose \emph{true} violation rate
(certified decisions that disagree with the exact ML class) is now
directly measurable on a BB code. We additionally compute the
\emph{global} ML class over all $2^{12}$ logical classes on a subsample
(mini-bucket screening plus exact evaluation of the survivors) to validate
the two-logical candidate restriction, sweep the region sizes
MBE($i$), $i\in\{8,12,16,20\}$, against the exact reference, and run
MBE($20$) with deterministic interval certificates on the gross code.

The exact reference changes what can be said. \emph{First}, exact
degenerate ML is now practical: with mini-bucket pruning ($110/120$ resp.\
$85/120$ syndromes decided by the intervals alone), exact decoding costs
$5$\,s ($p=0.03$) to $26$\,s ($p=0.05$) per syndrome, and at $p=0.03$ it
is strictly better than both BP+OSD variants ($0.0417$ versus $0.0500$):
BP+OSD is provably suboptimal on $2$ of $120$ syndromes there, while at
$p=0.05$ its errors happen to balance and the rates coincide.
\emph{Second}, the AIS decoder is validated at the strongest possible
standard: it agrees with exact ML on $238$ of $240$ syndromes, and among
its $231$ bootstrap-certified decisions there are \emph{zero} true
violations: the certificate's empirical error rate against exact ground
truth on a qLDPC code is $0/231$. The Bethe decoder agrees with exact ML
on $118/120$ at $p=0.03$ (uniform initialization) and degrades to
$87/120$ at $p=0.05$, quantifying precisely where the smallest regions
stop sufficing. \emph{Third}, the candidate restriction is validated
globally: on all tested syndromes the exact optimum over the full
$2^{12}$ classes (mini-bucket screening plus exact evaluation of
survivors) lies within the two-logical candidate set.
\emph{Fourth}, the region sweep: MBE($16$) and MBE($20$) reproduce the
exact decision on $40/40$ syndromes and \emph{deterministically certify}
$37$ resp.\ $38$ of them (machine-checkable optimality proofs at a
fraction of the exact cost), while MBE($8$) and MBE($12$) still decide
well ($38/40$, $26/40$) but certify nothing; the dip at $i=12$ is an
artifact of the greedy mini-bucket partition, whose cluster choices are
not monotone in $i$.

The gross code marks the hierarchy's current boundary, and we report it as
such: at induced width $\approx40$, the MBE($20$) upper bounds are loose
by far more than the class gaps, and ranking classes by them is
\emph{worse} than BP+OSD ($49$ of $60$ correct decisions broken, no
deterministic certificates, at $272$\,s per decode): a loose bound is not
a decision rule. Where the width exceeds the tractable cluster size by
this margin, the sampling decoder with its statistical certificate remains
the only degeneracy-aware option we can recommend; closing the gap with
tighter bounded-cluster bounds (weighted mini-buckets, join-graph
propagation) is a concrete open problem the coset MRF now poses to the
approximate-inference community.

\subsection{Certified Decoding of Hardware Syndromes: A Pilot}
\label{sec:exp6}
As a first step beyond simulation, we ran a pilot memory experiment on
IBM hardware (\texttt{ibm\_kingston}, Heron r2; median CZ error
$2.0\times10^{-3}$, readout $8\times10^{-3}$): the distance-3 rotated
surface code, memory-$Z$, three syndrome rounds with mid-circuit
measurement and reset, $10^4$ shots. The circuit is generated by stim
(fixing schedule, detectors, and the logical observable), transpiled to
the heavy-hex architecture (routing expands it to $26$ active qubits and
$273$ CZ gates), and converted back, gate by gate with per-gate
calibration errors attached, into a stabilizer noise model whose
flattened detector error model is the decoding problem ($413$ mechanisms,
$24$ detectors). Correctness of the round trip is verified by a
determinism check: the noiseless converted circuit fires no detector on
any sample. Decoding is performed under two models on identical data: the
\emph{calibration} model above, and a \emph{learned} model on the same
graph whose priors are re-fitted from the recorded detector statistics by
the standard two-point correlation estimator: MRF parameter estimation
from syndrome data, as anticipated in the outlook of earlier sections.

The pilot delivers three measurements, reported with the candor a pilot
deserves. \emph{First, the regime}: the routed circuit is far above
threshold: the measured mean detector rate is $0.229$ against the
calibration model's prediction of $0.076$, a factor-three mismatch
dominated by idle decoherence that the per-gate model omits. In this
regime decoding barely helps, as it cannot: the raw logical flip rate is
$0.253$ and the best decoder reaches $0.246$ (BP+OSD on $10^4$ shots;
AIS $0.260\pm0.025$ on a $300$-shot subsample, statistically
indistinguishable). \emph{Second, certificates on real data}: the AIS
decoder certifies $96\%$ of its decisions; correctly interpreted, these
assert ML-optimality \emph{under the assumed model}; the detector-rate
gap above, not the certificate, is what flags that the model itself is
off. \emph{Third}, re-fitting the priors from the recorded detector
statistics with the standard two-point estimator makes matters slightly
\emph{worse} ($0.253$), and the reason is structural: transpiled-circuit
error models are dominated by hyperedge mechanisms (three or more
detectors) that a pairwise estimator cannot see and must floor. Learning
hyperedge-rich detector error models from data, MRF parameter estimation
in earnest, is the concrete methodological gap this pilot isolates, and
the natural next step for the hardware program, together with an
idle-noise-aware calibration model.

\paragraph{First BB-code syndrome data.}
Finally, we ran the $[\![72,12,6]\!]$ code itself on hardware, to our
knowledge the first bivariate-bicycle syndrome data from any device. Full
syndrome cycles need degree-6 connectivity that no current processor
offers, but one round is well posed with only terminal measurements:
on \texttt{ibm\_berlin} (Nighthawk r1, 120-qubit square lattice, no
mid-circuit operations) we prepared $|0\rangle^{\otimes108}$ (72 data
qubits, 36 $Z$-ancillas), executed the six Tanner-monomial CNOT layers,
and measured everything once ($2\times10^4$ shots); final checks and all
twelve logical $\bar Z$'s are computed classically from the data readout.
The routed circuit (1422 CZ gates on dead-coupler-free edges; both
determinism audits pass) is deep, and the data sit near saturation: mean
detector rate $0.355$ (calibration model: $0.238$), raw per-logical flip
rate $0.439$. Decoding still extracts real structure (per-logical error
drops to $0.349$ and the block error rate from $0.983$ to $0.709$, the
large block-level gain revealing strongly \emph{correlated} logical
failures), and the AIS decoder certifies $100\%$ of its decisions while
agreeing with spacetime BP+OSD on all $200$ decoded syndromes ($8$\,s per
decode, 695-mechanism model). The experiment is a decoding study, not
error suppression: it establishes that the certified pipeline runs
end-to-end on real qLDPC syndrome data, and it marks the connectivity
frontier where hardware, not decoding, is the binding constraint.

\section{Discussion and Outlook}
\label{sec:discussion}
Degenerate maximum-likelihood decoding of CSS codes is partition-function
estimation: not metaphorically, but by the exact reduction of
Proposition~\ref{prop:cosetmrf}, which turns each logical class into an
unconstrained, strictly positive Markov random field over the code's check
variables and extends unchanged to spacetime decoding with faulty
measurements. This reframing makes three decades of
probabilistic-inference machinery directly applicable to the decoding
problem of the code families most likely to power early fault-tolerant
machines, and our experiments show that both of its main branches pay off.
The \emph{sampling} branch yields a decoder with a property no standard
qLDPC decoder offers: per-decision certificates, either statistical
(paired bootstrap on CRN-coupled anneals) or exact (Lemma~\ref{lem:cert}
composed with constant-factor estimators such as WISH). It tracks the
exact degenerate ML decoder on surface codes, matches the BP+OSD family on
bivariate bicycle codes, and certifies the bulk of its decisions from code
capacity through circuit-level noise, validated on the
$[\![72,12,6]\!]$ code against an exact reference, with $0$ true
violations among $231$ certified decisions. The \emph{region} branch
spans the hierarchy: the Bethe decoder reproduces exact ML on surface
codes at millisecond cost (a degeneracy-aware decoder in the latency
class of MWPM), and one level up, elimination clusters turn the modest
induced width of the $[\![72,12,6]\!]$ coset graph into the first exact
degenerate-ML decoder for a qLDPC code, strictly outperforming BP+OSD and
supplying deterministic mini-bucket certificates; the same construction's
failure at gross-code width ($49/60$ decisions broken by loose
MBE bounds) marks precisely where tighter bounded-cluster bounds are
needed. Equally informative are the failures we report: direct inter-class
bridging and BAR collapse on the low overlap between coset distributions,
establishing that the anneal through infinite temperature is the right
bridge and that certificate strength comes from the test statistic, not
from a shorter path.

The obvious open front for the sampling branch is speed: seconds per
decode versus the microsecond real-time loop of superconducting hardware.
We see three roles for it nonetheless. First, as a \emph{reference
decoder}: certified ML decisions provide the yardstick against which fast
heuristics should be measured on qLDPC codes, a role so far playable only
on the surface code; our Bethe and BP+OSD comparisons already use it this
way. Second, as a \emph{second-stage triage decoder}: the certificate
marks the small fraction of syndromes on which fast decoders are likely
wrong, and slower certified decoding of only those syndromes fits latency
budgets in memory experiments and offline analysis. Third, as a bridge to
\emph{annealing hardware}: WISH needs only MAP oracles, each a QUBO on the
coset MRF, so a quantum annealer can in principle serve as the inner loop
of a certified decoder, a code-hardware pairing complementary to the
platform discussion of~\cite{bravyi2024, xu2024}. The Bethe branch invites
the opposite program: it is fast enough for real-time use, and its two
open problems, fixed-point selection at strong coupling and
Kikuchi regions absorbing the CSS four-cycles, are classical
approximate-inference questions with a large existing
literature~\cite{yedidia2005}. Beyond decoding, the coset MRF makes
learned proposal distributions and parameter estimation of correlated
noise models from syndrome streams standard inference questions on a
well-defined model class.

\section*{Acknowledgments}
This research has been funded by the Federal Ministry of Research,
Technology and Space of Germany (BMFTR) and the state of North
Rhine-Westphalia as part of the Lamarr Institute for Machine Learning
and Artificial Intelligence.

\section*{Author contributions}
N.P.\ conceived the idea and wrote the initial draft. All authors
contributed to the writing and to the discussion of the results. Large
language models were used to optimize the text.

\bibliographystyle{unsrt}
\bibliography{references}

\end{document}